\documentclass[10pt]{amsart}
\usepackage[english]{babel}
\usepackage{amssymb}
\usepackage{amsmath}
\usepackage{srcltx}

\newcommand{\tl}{\widetilde{\tau}}

\newcommand{\ignore}[1]{}

\newtheorem{dummy}{Dummy}

\newtheorem{lemma}[dummy]{Lemma}
\newtheorem{theorem}[dummy]{Theorem}
\newtheorem{proposition}[dummy]{Proposition}
\newtheorem{corollary}[dummy]{Corollary}

\theoremstyle{definition}
\newtheorem{definition}{Definition}

\newtheorem{example}[dummy]{Example}

\newtheorem{remark}[dummy]{Remark}

\title[Nonassociative algebras used to build space-time block codes]
{The nonassociative algebras used to build fast-decodable space-time block codes}

\author[Susanne Pumpl\"un]{}
\author[A. Steele]{}



\author{S. Pumpl\"un}
\author{A. Steele}
\email{susanne.pumpluen@nottingham.ac.uk;
pmxas4@nottingham.ac.uk}

\address{School of Mathematical Sciences\\
University of Nottingham\\
University Park\\
Nottingham NG7 2RD\\
United Kingdom
}

\address{Flat 203, Wilson Tower\\
16 Christian Street\\
London E1 1AW\\
United Kingdom
}

\begin{document}
\maketitle

\begin{abstract}
Let $K/F$ and $K/L$ be two  cyclic Galois field extensions  and $D=(K/F,\sigma,c)$ a cyclic algebra. Given an
invertible element $d\in D$,
we present three families of unital nonassociative algebras over $L\cap F$ defined
on the direct sum of $n$ copies of
 $D$.
Two of these families appear either explicitly or implicitly in the designs of fast-decodable space-time block codes
in  papers
by Srinath, Rajan, Markin, Oggier, and the authors.
We present conditions for the algebras to be division and propose a construction for fully diverse
 fast decodable space-time block codes of rate-$m$ for $nm$ transmit and $m$ receive antennas.
We present a  DMT-optimal rate-3 code for 6 transmit and 3 receive antennas which is fast-decodable,
with ML-decoding complexity
at most $\mathcal{O}(M^{15})$.
\end{abstract}

\maketitle

\section{Introduction}

Space-time block codes (STBCs) are used for reliable high rate transmission
over wireless digital channels with multiple antennas at both the transmitter and receiver ends.
From the mathematical point of view, a space-time block code is a set of complex $n\times m$ matrices,
the codebook, that satisfies a number of properties which determine how well the code performs.

Recently, several different constructions of nonassociative algebras appeared in the literature on fast decodable
STBCs, cf. for instance Markin and Oggier \cite{MO13}, Srinath and Rajan \cite{R13}, or
 \cite{SP14}, \cite{P13.2}, \cite{SPO12}, \cite{PU11}.
 There are two different types of algebras involved.
 The aim of this paper is to present them in a unified manner and
 investigate their structure, in order to be able to build the
  associated (fully diverse, fast-decodable) codes more efficiently in the future.

Let $K/L$ be a cyclic Galois field extension with Galois group ${\rm Gal}(K/L)=\langle \tau\rangle$ of degree $n$
and $K/F$ a cyclic Galois field extension with Galois group ${\rm Gal}(K/F)=\langle \sigma\rangle$ of degree $m$.
Put $F_0=F\cap L$.
Given the direct sum $A$ of $n$ copies of a cyclic  algebra $D=(K/F,\sigma,c)$, $c\in F_0$, we define three
different multiplications on $A$, which each turn $A$ into a unital nonassociative algebra over $F_0$.
We canonically extend $\tau$ to an $L$-linear map $\widetilde{\tau}:D\to D$,
 choose an element $d\in D^\times$ and define a multiplication on the right $D$-module
\[ D \oplus fD \oplus f^2D \oplus \cdots \oplus f^{n-1}D\]
 via
\[
 (f^i x)(f^j y) =
  \begin{cases}
   f^{i+j} \tl^j(x)y & \text{if } i+j < n, \\
   f^{(i+j)-n} d \tl^j(x)y & \text{if } i+j \geq n,
  \end{cases}
\]
\[
 (f^i x)(f^j y) =
  \begin{cases}
   f^{i+j} \tl^j(x)y & \text{if } i+j < n, \\
   f^{(i+j)-n}  \tl^j(x) d y  & \text{if } i+j \geq n,
  \end{cases}
\]
or
\[
 (f^i x)(f^j y) =
  \begin{cases}
   f^{i+j} \tl^j(x)y & \text{if } i+j < n, \\
   f^{(i+j)-n} \tl^j(x)yd  & \text{if } i+j \geq n
  \end{cases}
\]
for all $x, y \in D$, $0\leq i,j<n$. We call the resulting algebra ${\rm It}^n(D,\tau,d)$, ${\rm It}_M^n(D,\tau,d)$ or
 $It_R^n(D, \tau, d) $, respectively.

For $A={\rm It}^n(D,\tau,d)$ and $A={\rm It}^n_M(D,\tau,d)$, the left multiplication $L_x$ with a non-zero
element $x\in A$ can be represented by
an $ nm\times nm$ matrix with entries in $K$
 (considering $A$ as a right $K$-vector space of dimension $mn$).

 For $d\in L^\times$,  left multiplication
 $L_x$ with a non-zero element $x\in A=It_R^n(D,\tau,d)$ is a $K$-endomorphism as well, and can be represented
by an $nm\times nm$ matrix with entries in $K$.

 The family of matrices representing left multiplication in any of the three cases can be used to define
a STBC $\mathcal{C}$, which is fully diverse if and only if $A$ is division, and
 fast-decodable for the right choice of $D$.

 The three algebra constructions in this paper generalize the three types of {\it iterated algebras} presented in
\cite{P13.2} (the $n=2$ case).
 A first question concerning their existence
 can be found in Section VI. of  \cite{MO13}; the iterated codes treated there arise from the algebra
 ${\rm It}^2(D,\tau,d)$.
 The algebras ${\rm It}^n(D,\tau,d)$ and ${\rm It}^n_R(D,\tau,d)$
 appear when  designing fast-decodable asymmetric multiple input double output (MIDO) codes:
 ${\rm It}^n_R(D,\tau,d)$ is implicitly used in \cite{R13} but not mentioned there,
the algebras ${\rm It}^n(D,\tau,d)$ are canonical generalizations of the ones behind the iterated codes of  \cite{MO13},
and are employed in \cite{SP14}. Both times they
 are used to design fast decodable rate-2 MIDO
space-time block codes with $n$ antennas transmitting and $2$ antennas receiving the data.
All  codes for $n>2$ transmit antennas presented in
\cite{R13} and all but one \cite{SP14} have sparse entries and therefore do not have a high data rate.

We include the third family, ${\rm It}^n_M(D,\tau,d)$, for completeness.

After the preliminaries in Section 2,
 the algebras ${\rm It}^n(D,\tau,d)$ and ${\rm It}_M^n(D,\tau,d)$ are investigated in Section 3.
 Several necessary and sufficient conditions for
 ${\rm It}^n(D,\tau,d)$ to be a division algebra are given if $d\in F^\times$. For instance,
if $n$ is prime and in case $n\not=2,3$, additionally $F_0$ contains a primitive $n$th root of unity, then
 ${\rm It}^n(D,\tau,d)$ is a division algebra  for all $d\in F\setminus F_0$ with
$d^m\not\in F_0$ (Proposition \ref{prop:6}).
 Section 4 deals with
the algebras ${\rm It}_R^n(D,\tau,d)$ which were defined by B. S. Rajan and L. P. Natarajan
(and for $d\in L\setminus F$ yield the codes  in~\cite{R13}).
They were already
defined previously in a little known paper by Petit \cite{P66} using twisted polynomial rings.
Necessary and sufficient conditions for ${\rm It}_R^n(D,\tau,d)$ to be a division algebra are given and simplified for
special cases.
E.g., if $D$ is a quaternion division algebra, ${\rm It}_R^3(D,\tau,d)$ is a division algebra
for all $d\in L\setminus F$ with $d\not\in N_{K/L}(K^\times)$ (Theorem \ref{thm:degree3}).

Some of these conditions are simplification of the ones contained in an earlier version of
 this paper, applied in \cite{SP14} when designing fully diverse codes. In particular for the case $n=3$,
 Proposition \ref{prop:6} makes it easy now to build fully diverse codes of maximal rate using ${\rm It}^3(D,\tau,d)$
 and Theorem \ref{thm:degree3} using ${\rm It}_R^3(D,\tau,d)$.
 Previously, there were no criteria known to check such iterated $6\times 3$-codes for full inversibility.

How to design fully diverse fast-decodable  multiple input multiple output (MIMO) codes
for $nm$ transmit and $m$ receive antennas employing certain
${\rm It}_R^n(D,\tau,d)$ and ${\rm It}^n(D,\tau,d)$ is explained in Sections 5 and 6:  if the code
associated to $D$ is fast-decodable, then so is the one associated to ${\rm It}_R^n(D,\tau,d)$, respectively,
${\rm It}^n(D,\tau,d)$.
We are interested in a high data rate and use the $mn^2$ degrees of freedom of the channel to transmit
$mn^2$ complex symbols.
 Our method yields fully diverse
codes of rate-$m$ for $nm$ transmit and $m$ receive antennas, which is maximal rate for $m$ receive antennas.
We present two examples of a rate-3 code for 6 transmit and 3 receive antennas which are fast-decodable with
 ML-decoding complexity at most $\mathcal{O}(M^{15})$ (using the M-HEX constellation).
One of them is DMT-optimal
and has normalized minimum determinant $49(\frac{2}{\sqrt{28E}})^{18}=1/7^7 E^9$.
We also give an example of a rate-4 code for 8 transmit and 4 receive antennas which is fast-decodable with
ML-decoding complexity
at most $\mathcal{O}(M^{26})$ (using the M-QAM constellation).
The suggested codes  have maximal rate in terms of the number of complex symbols per channel use (cspcu).

\section{Preliminaries}

\subsection{Nonassociative algebras}

Let $F$ be a field. By ``$F$-algebra'' we mean a finite dimensional nonassociative algebra over $F$ with unit element $1$.

 A nonassociative algebra $A\not=0$ is called a {\it division algebra} if for any $a\in A$, $a\not=0$,
the left multiplication  with $a$, $L_a(x)=ax$,  and the right multiplication with $a$, $R_a(x)=xa$, are bijective.
$A$ is a division algebra if and only if $A$ has no zero divisors \cite[pp. 15, 16]{Sch2}.

For an $F$-algebra $A$, associativity in $A$ is measured by the {\it associator} $[x, y, z] = (xy) z - x (yz)$.
The
{\it middle nucleus} of $A$ is
defined as ${\rm Nuc}_m(A) = \{ x \in A \, \vert \, [A, x, A]  = 0 \}$ and  the
{\it  nucleus} of $A$ is
defined as  ${\rm Nuc}(A) = \{ x \in A \, \vert \, [x, A, A] = [A, x, A] = [A,A, x] = 0 \}$.
The nucleus is an associative
subalgebra of $A$ containing $F1$ and $x(yz) = (xy) z$ whenever one of the elements $x, y, z$ is in
${\rm Nuc}(A)$.
The {\it commuter} of $A$ is defined as ${\rm Comm}(A)=\{x\in A\,|\,xy=yx \text{ for all }y\in A\}$
and the {\it center} of $A$ is ${\rm C}(A)=\{x\in A\,|\, x\in \text{Nuc}(A) \text{ and }xy=yx \text{ for all }y\in A\}$.

For coding purposes, often algebras are considered as a vector space over some
 subfield $K$, $F\subset K\subset A$.
  Usually $K$ is maximal with respect to inclusion. For nonassociative algebras, this is for instance
   possible if $K \subset \text{Nuc}(A) $.

  If then left multiplication $L_x$ is a $K$-linear map
for an algebra $A$ over $F$ we can consider the map
$$\lambda: A \to {\rm End}_K(A), x\mapsto L_x$$
which induces a map
$$\lambda: A \to {\rm Mat}_s(K), x\mapsto L_x \mapsto \lambda(x)= X$$
with $s=[A: K]$,
after choosing a $K$-basis for $A$ and expressing the endomorphism $L_x$ in matrix form.
For an associative algebra, this is the left regular representation of $A$.

If $A$ is a division algebra, $\lambda$ is an embedding of vector spaces.

Similarly, given an associative subalgebra $D$ of $A$ such that $A$ is a free right $D$-module and
such that  left multiplication $L_x$ is a right $D$-module endomorphism,
we can consider the map
$$\lambda: A \to {\rm End}_D(A), x\mapsto L_x$$
which induces a map
$$\lambda: A \to {\rm Mat}_t(D), x\mapsto L_x \mapsto \lambda(x)=X$$
with $t={\rm dim}_D A$,
after choosing a $D$-basis for $A$.

\subsection{Associative and nonassociative cyclic algebras}

Let $K/F$ be a cyclic Galois extension of degree $m$, with Galois group ${\rm Gal}(K/F)=\langle \sigma \rangle$.

Let $c\in F^\times$.
An {\it associative cyclic algebra} $A=(K/F,\sigma,c)$ {\it of degree} $m$ over $F$ is an $m$-dimensional $K$-vector space
$
A=K \oplus eK \oplus e^2 K\oplus\dots \oplus e^{m-1}K,
$
with multiplication given by the relations
$$\label{eq:rule}
e^m=c,~xe=e\sigma(x),
$$
for all $x\in K$.  If  $c^s \neq N_{K/L}(x)$ for all $x \in K$ and all $1 \leq s \leq m-1$, then $A$ is a division algebra.

For any $c\in K\backslash F$, the {\it nonassociative cyclic algebra}  $A=(K/F,\sigma,c)$
{\it of degree} $m$ is
given by the $m$-dimensional $K$-vector space
$
A=K \oplus eK \oplus e^2K \oplus \dots\oplus e^{m-1}K
$
together with the rules
\[
 (e^i x)(e^j y) =
  \begin{cases}
   e^{i+j} \sigma^j(x)y & \text{if } i+j < m \\
   e^{(i+j)-m} c \sigma^j(x)y & \text{if } i+j \geq m
  \end{cases}
\]
for all $x,y\in K, 0 \leq i,j, <m$, which are extended linearly to all elements of $A$ to define the multiplication of $A$.

The unital algebra $(K/F,\sigma,c)$, $c\in K\setminus F$  is not $(n+1)$st power associative, but is built similar to the
 associative cyclic algebra $(K/F,\sigma,c)$, where $c\in F^\times$:
 we again have
 $$xe=e\sigma(x) \text{ and } e^ie^j=c$$
for all integers $i,j$ such that $i+j=m$, so that $e^m$ is well-defined and $e^m=c.$
$(K/F,\sigma,c)$ has nucleus  $K$  and center $F$.
 If $c \in K \setminus F$  is such that $1, c, c^2, \ldots, c^{m-1}$ are linearly independent over $F$, then
  $A$ is a division algebra. In particular, if $m$ is prime, then $A$ is division for any choice of $c \in K \setminus F$.
Nonassociative cyclic algebras are studied extensively in \cite{S12}.

\subsection{Iterated algebras \cite{P13.2} }

Let $K/F$ be a cyclic Galois extension of degree $m$ with Galois group
${\rm Gal}(K/F)=\langle\sigma\rangle$ and $\tau\in {\rm Aut} (K)$.
Define $L={\rm Fix}(\tau)$ and $F_0=L\cap F$.
Let $D=(K/F,\sigma,c)$ be an associative cyclic algebra over $F$ of degree $m$.
 For $x= x_0 + ex_1 + e^2x_2 +\dots + e^{m-1}x_{m-1}\in D$, define the $L$-linear map
 $\widetilde{\tau}:D\to D$  via
$$\widetilde{\tau}(x)=\tau(x_0) + e \tau(x_1) + e^2\tau(x_2) +\dots + e^{m-1}\tau(x_{m-1}).$$
If $\tau^m=id$ then $\widetilde{\tau}^m=id$.

\begin{remark}
 Let $c\in L$.
 \\
(i) $\tl(xy) = \tl(x)\tl(y)$  and $\lambda (\tl(x)) = \tau (\lambda(x))$ for all $x,y \in D$,
where for any matrix $X=\lambda(x)$ representing left multiplication with $x$, $\tau(X)$ means applying $\tau$ to each entry of the matrix.
\\ (ii)
Let $D'=(K/F, \sigma,\tau(c))$ with standard basis $1,e',\dots,{e'}^{m-1}$.
For
$y=y_0+ey_1+\dots+e^{m-1}y_{m-1}\in D$ define
$y_{D'}=y_0+e'y_1+\dots+{e'}^{m-1}y_{m-1}\in D'$.
 By \cite[Proposition 4]{P13.2},
$N_{D/F}(\widetilde{\tau}(y))=\tau(N_{D/F}(y)).$
\end{remark}

Choose $d\in D^\times$.
 Then the $2m^2$-dimensional $F$-vector space
 $A=D\oplus D$ can be made into a unital algebra over $F_0$ via the multiplication
$$(u,v)(u',v')=(uu'+d \widetilde{\tau}(v) v',vu'+ \widetilde{\tau}( u)v'),$$
$$(u,v)(u',v')=(uu'+ \widetilde{\tau}(v) d v',vu'+ \widetilde{\tau}( u)v')$$
resp.
$$(u,v)(u',v')=(uu'+ \widetilde{\tau}(v) v' d,vu'+ \widetilde{\tau}( u)v')$$
for $u,u',v,v'\in D$ with unit element $1=(1_D,0)$.
The corresponding algebras are denoted by
${\rm It}(D,\tau,d)$, ${\rm It}_M(D,\tau,d)$, resp. ${\rm It}_R(D, \tau, d)$, and have dimension $2m^2[F:F_0]$ over $F_0$.
${\rm It}(D,\tau,d)$, ${\rm It}_M(D,\tau,d)$ and ${\rm It}_R(D,\tau,d)$ are called {\it iterated algebras} over $F$.

Every iterated algebra $A$ as above is a right $D$-modules
with  $D$-basis $\{1,f\}$.
We can therefore embed ${\rm End}_D(A)$ into the module
${\rm Mat}_2(D)$.
Furthermore, for $A={\rm It}(D,\tau,d)$
and $A={\rm It}_M(D,\tau,d)$ left multiplication
$L_x$ with $x\in A$ is a $D$-linear map, so that we have a well-defined  additive map
$$L:A\to {\rm End}_D(A)\subset {\rm Mat}_2(D),\quad x\mapsto L_x,$$
which is injective if $A$ is division. $L_x$ can also be viewed as a $K$-linear map
and
after a choice of  $K$-basis for $A$, we can embed ${\rm End}_K(A)$ into the vector space
${\rm Mat}_{2m}(K)$ via $\lambda: A\to {\rm Mat}_{2m}(K),$ $x \mapsto L_x$.
By restricting to $d\in L^\times$, we
 achieve that left multiplication $L_x$ in ${\rm It}_R(D,\tau,d)$ is a $K$-endomorphism and thus also can be represented
by a matrix with entries in $K$, as for the two other algebras.
Therefore if $d\in L^\times$, we can embed ${\rm End}_K(A)$ into the vector space
${\rm Mat}_{2m}(K)$ via $\lambda: A\to {\rm Mat}_{2m}(K),$ $x \mapsto L_x$ for $A={\rm It}_R(D,\tau,d)$ as well.

\begin{theorem}  \label{thm:oldmain} (\cite[Theorem 3.2]{P13.2}, \cite[Theorem 1]{R13})
 Let $D$ be a cyclic division algebra of degree $n$ over $F$ with norm $N_{D/F}$  and $d\in D^\times$.
Let $\tau\in {\rm Aut} (K)$ and suppose $\tau$ commutes with $\sigma$.
 Let $A={\rm It}(D,\tau,d)$, $A={\rm It}_M(D,\tau,d)$ or $A={\rm It}_R(D,\tau,d)$.
\\ (i) $A$ is a division algebra if
$$N_{D/F}(d)\not= N_{D/F}(z\widetilde{\tau}(z))$$
for all $z\in D$.
Conversely, if $A$ is a division algebra then
$d\not=z\widetilde{\tau}(z)$ for all $z\in D^\times$.
\\ (ii)  Suppose $c\in {\rm Fix}(\tau)$. Then:
\\ (a) $A$ is a division algebra if and only if  $d\not=z\widetilde{\tau}(z)$ for all $z\in D$.
\\ (b) $A$ is a division algebra if $N_{D/F}(d)\not=a \tau(a)$ for all $a\in N_{D/F}(D^\times)$.
\\ (iii)  Suppose $F\subset {\rm Fix}(\tau)$. Then $A$
 is a division algebra if $N_{D/F}(d)\not\in N_{D/F}(D^\times)^2$.
\end{theorem}

\subsection{Design criteria for space-time block codes}\label{sec:def}

A space-time block code (STBC) for an $n_t$ transmit antenna MIMO system is a set of complex $n_t\times T$ matrices,
called codebook, that satisfies a number of properties which determine how well the code performs.
Here, $n_t$ is the number
 of transmitting antennas, $T$ the number of channels used.

Most of the existing  codes  are built from cyclic division algebras
over number fields $F$, in particular over $F=\mathbb{Q}(i)$ or $F= \mathbb{Q}(\omega)$
with $\omega=e^{2\pi i /3}$ a third root of unity, since these fields are used for
 the transmission of QAM or HEX constellations, respectively.

One goal is to find \emph{fully diverse} codebooks $\mathcal{C} $, where the difference of any two
code words   has full rank, i.e. with $\det(X-X')\not=0$ for all matrices $X\not=X',$ $X,X'\in \mathcal{C} $.

 If the minimum determinant of the code, defined as
$$\delta(\mathcal{C})=\inf_{X'\not=X''\in \mathcal{C}}|\det(X'-X'')|^2,$$
 is bounded below by a constant, even if the codebook $\mathcal{C}$   is infinite,
 the code $\mathcal{C} $ has \emph{non-vanishing determinant} (NVD).
 Since our codebooks $\mathcal{C}$ are based on the matrix representing left multiplication in
 an algebra, they are linear and thus their
 minimum determinant is given by
$$\delta(\mathcal{C})=\inf_{0\not=X\in \mathcal{C}}|\det(X)|^2.$$
 If $\mathcal{C}$ is fully diverse,
$\delta(\mathcal{C})$ defines the \emph{coding gain} $\delta(\mathcal{C})^{\frac{1}{n_t}}$.
The larger $\delta(\mathcal{C})$ is, the better the error performance of the code is expected to be.

If a STBC has NVD then it will perform well independently of the constellation size we choose.
The NVD property guarantees that a full rate linear STBC  has
optimal diversity-multiplexing gain trade-off (DMT) and
 also an asymmetric linear STBC with NVD often has
DMT
(for results on the relation between NVD and DMT-optimality for asymmetric linear STBCs, cf. for instance \cite{SR3}).

We  look at transmission over a MIMO fading channel with $n_t=nm$ transmit and $n$ receive antennas,
and assume the channel is coherent, that is the receiver has perfect knowledge of the channel.
We consider the  rate-$n$ case (where $mn^2$ symbols are sent).
The system is modeled as
$$Y=\sqrt{\rho}HS+N,$$
with $Y$ the complex $n_r\times T$ matrix consisting of the received signals, $S$ the
the complex $n_t\times T$ codeword matrix, $H$ is the the complex $n_r\times n_t$ channel matrix (which we assume to be
  known)  and $N$
the the complex $n_r\times T$ noise matrix, their entries being identically independently distributed Gaussian random variables with mean zero and
variance one. $\rho$ is the average signal to noise ratio.

 Since we assume the channel is coherent,
ML-decoding can be obtained via sphere decoding. The hope is to find codes which are easy to decode with a sphere decoder,
i.e. which are fast-decodable:
Let $M$ be the size of a complex constellation of coding symbols and assume the code $\mathcal{C} $ encodes $s$ symbols.
 If the decoding complexity by sphere decoder needs only $\mathcal{O}(M^l)$, $l<s$ computations, then
$\mathcal{C} $ is called \emph{fast-decodable}.

For a matrix $B$, let $B^*$ denote its Hermitian transpose. Consider a code $\mathcal{C}$ of rate $n$. Any
$X\in\mathcal{C}\subset {\rm Mat}_{mn\times mn}(\mathbb{C})$ can be written as a linear combination
$$X=\sum_{i=1}^{nm^2}g_iB_i,$$
of $nm^2$ $\mathbb{R}$-linearly independent basis matrices $B_1,\dots,B_{nm^2}$,
with $g_i\in\mathbb{R}$.
 Define
$$M_{g,k}=||B_g B^*_k+B_kB^*_g||.$$
Let $S$ be a real constellation of coding symbols.
A STBC with $s=nm^2$ linear independent real information symbols from $S$ in one code matrix is
called \emph{$l$-group decodable}, if there is a partition of $\{1,\dots,s\}$ into $l$ nonempty subsets
$\Gamma_1,\dots,\Gamma_l$, so that $M_{g,k}=0$, where $g,k$ lie in disjoint subsets $\Gamma_i,\dots,\Gamma_j$.
The code $\mathcal{C}$ then has decoding complexity $\mathcal{O}(|S|^{L})$, where $L=max_{1\leq i\leq l}|\Gamma_i|$.

\section{General iteration processes I and II}

We will use the notation defined below throughout the remainder of the paper:
Let $F$ and $L$ be fields and let $K$ be a cyclic extension of both $F$ and $L$ such that
\begin{enumerate}
\item $Gal(K/F) = \langle \sigma \rangle$ and $[K:F] = m$,
\item $Gal(K/L) = \langle \tau \rangle$ and $[K:L] = n$,
\item $\sigma$ and $\tau$ commute: $\sigma \tau = \tau \sigma$.
\end{enumerate}
Let $F_0=F\cap L$. Let $D=(K/F, \sigma, c)$ be an associative cyclic division algebra over $F$ of degree $m$ with norm
$N_{D/F}$  and $c\in F_0$. The condition that $c \in F_0$ means that $\widetilde{\tau}\in {\rm Aut}_{F_0}(D)$ of order $n$,
 see the definition of  $\widetilde{\tau}$ in Section 2.3.

\begin{definition} \label{maindef}
 Pick $d \in D^\times$. Define a multiplication on the right $D$-module
$ D \oplus fD \oplus f^2D \oplus \cdots \oplus f^{n-1}D$
 via
 \\ (i)
\[
 (f^i x)(f^j y) =
  \begin{cases}
   f^{i+j} \tl^j(x)y & \text{if } i+j < n \\
   f^{(i+j)-n} d \tl^j(x)y & \text{if } i+j \geq n
  \end{cases}
\]
for all $x, y \in D$, $i,j<n$, and call the resulting algebra ${\rm It}^n(D,\tau,d)$, or via
\\ (ii)
\[
 (f^i x)(f^j y) =
  \begin{cases}
   f^{i+j} \tl^j(x)y & \text{if } i+j < n \\
   f^{(i+j)-n}  \tl^j(x) d y  & \text{if } i+j \geq n
  \end{cases}
\]
for all $x, y \in D$, $i,j<n$, and call the resulting algebra ${\rm It}_M^n(D,\tau,d)$.
\end{definition}

 ${\rm It}^n(D,\tau,d)$ and ${\rm It}_M^n(D,\tau,d)$ are both nonassociative algebras over $F_0$ of dimension $nm^2[F:F_0]$
 with unit element $1\in D$ and contain $D$ as a subalgebra.
For both, $f^{n-1}f=d=ff^{n-1}.$
 If $d\in F^\times$ then ${\rm It}^n(D,\tau,d)={\rm It}_M^n(D,\tau,d)$.

Moreover, ${\rm It}^2(D,\tau,d)={\rm It}(D,\tau,d)$ and ${\rm It}_M^2(D,\tau,d)={\rm It}_M(D,\tau,d)$ are the
 iterated algebras from Section 2.3.
The algebras ${\rm It}^n(D,\tau,d)$ are canonical generalizations of the ones behind the iterated codes of  \cite{MO13},
 and employed in \cite{SP14}.

Let $A$ be either ${\rm It}^n(D, \tau, d)$ or ${\rm It}^n_M(D, \tau, d)$, unless specified differently.

\begin{lemma}\label{lem:lem3}
 (i) If $d\in K^\times$, then  $(K/L,\tau,d)$, viewed as an algebra over $F_0$,
 is a subalgebra of $A$.
 If $d\in L^\times$, then $(K/L,\tau,d)$ is an associative cyclic algebra of degree $n$, if $d\in K\setminus L$, $(K/L,\tau,d)$
 is a nonassociative cyclic algebra  of degree $n$.
\\ (ii) $A\otimes_{F_0} K={\rm Mat}_m(K)\oplus f{\rm Mat}_m(K)\oplus\dots\oplus f^{n-1}{\rm Mat}_m(K)$
 contains the $F_0$-algebra
${\rm Mat}_m(K)$ as a subalgebra and has zero divisors.
\\ If $d\in L^\times$ then $A\otimes_{F_0} K$ also contains the $F_0$-algebra ${\rm Mat}_n(K)$ as a subalgebra.
\\ (iii) Let $n=2s$  for some integer $s$.
Then  ${\rm It}(D, \tau^s, d)$ (resp. ${\rm It}_M(D, \tau^s, d)$) is isomorphic to a subalgebra of
 ${\rm It}^n(D, \tau, d)$ (resp. ${\rm It}^n_M(D, \tau, d)$).
\\ (iv) $D$ is contained in the middle nucleus of ${\rm It}^n(D, \tau, d)$.
\end{lemma}

\begin{proof}
(i) Restricting the multiplication of $A$ to entries in $K$ proves the assertion immediately:
By slight abuse of notation, we have ${\rm It}^n(K,\tau,d)=(K/L,\tau,d).$
\\ (ii)  is trivial as $D\otimes_{F_0} K\cong {\rm Mat}_m(K)$ splits. If $d\in L^\times$ then $A$ has the $F_0$-subalgebra
$(K/L,\tau,d)$, which as an algebra has splitting field $K$.
\\ (iii)  It is straightforward to check that $A$ is isomorphic to $D \oplus f^sD$, which is a subalgebra of $A$
under the multiplication inherited from $A$.
\\
(iv) By linearity of multiplication, we only need to show that
\[((f^ix)y)f^jz = f^ix(y(f^jz)),\]
for all $x,y,z \in D$ and all integers $0 \leq i,j \leq n-1$.
 A straightforward calculation shows that these are equal if and only if
 $\tl(x)\tl(y) = \tl(xy)$ for all $x,y \in D$. This is true if and only if $\tau(c) = c$.
\end{proof}

Lemma \ref{lem:lem3} (iii) can be generalized to the case where $n$ is any composite number if needed.

 $A$ is a free right $D$-module of rank $n$,
with right $D$-basis
 $\{1,f,\dots,f^{n-1}\}$
and we can embed ${\rm End}_D(A)$ into  ${\rm Mat}_n(D)$.
Left multiplication $L_x$ with $x\in A$
 is a right $D$-endomorphism, so that we obtain a well-defined  additive map
$$\lambda:A\to  {\rm Mat}_n(D),\quad x \mapsto L_x.$$
Let $x,y\in A$, $x = x_0 + fx_1 + f^2 x_2 +\cdots + f^{n-1}x_{n-1}$,
$y = y_0 + fy_1 + \cdots f^{n-1}y_{n-1}$ with $x_i,y_i\in D$.
 If we represent $y$ as a column vector $(y_0, y_1, \ldots, y_{n-1})^T$,
  then we can write the product of $x$ and $y$ in  $ A$ as a matrix multiplication
\[ xy = M(x)y,\]
where $M(x)$ is an $n \times n$ matrix with entries in $D$ given by
\[M(x) = \left[ \begin{array}{ccccc}
x_0 & d \tl(x_{n-1})& d \tl^2(x_{m-2}) & \cdots & d \tl^{n-1}(x_1) \\
x_1 & \tl(x_0) & d \tl^2(x_{n-1}) & \cdots & d \tl^{n-1}(x_{2}) \\
x_2 & \tl(x_1) & \tl^2(x_0) & \cdots & d \tl^{n-1}(x_3)\\
\vdots & \vdots & \vdots & \ddots & \vdots \\
x_{n-1} & \tl(x_{n-2}) & \tl^2(x_{n-3}) & \cdots & \tl^{n-1}(x_0) \end{array} \right] \]
if $A={\rm It}^n(D, \tau, d)$ and
\[M(x) = \left[ \begin{array}{ccccc}
x_0 &  \tl(x_{n-1}) d&  \tl^2(x_{m-2}) d & \cdots &  \tl^{n-1}(x_1) d \\
x_1 & \tl(x_0) &  \tl^2(x_{n-1}) d & \cdots &  \tl^{n-1}(x_{2}) d\\
x_2 & \tl(x_1) & \tl^2(x_0) & \cdots &  \tl^{n-1}(x_3) d\\
\vdots & \vdots & \vdots & \ddots & \vdots \\
x_{n-1} & \tl(x_{n-2}) & \tl^2(x_{n-3}) & \cdots & \tl^{n-1}(x_0) \end{array} \right] \]
if $A={\rm It}_M^n(D, \tau, d)$.

\begin{example}
Let $A={\rm It}^3(D,\tau,d)$ or $A={\rm It}_M^3(D,\tau,d)$ with $d\in D$.
For $f=(0,1,0)$, we have $f^2=(0,0,1)$ and $f^2f=(d,0,0)=ff^2.$
 The multiplication in ${\rm It}^3(D,\tau,d)$ is given by
\[(u,v,w)(u',v',w')
=(
\begin{bmatrix}
u & d\widetilde{\tau}(w) & d\widetilde{\tau}^2(v)  \\
v & \widetilde{\tau}(u) & d\widetilde{\tau}^2(w)  \\
w & \widetilde{\tau}(v) & \widetilde{\tau}^2(u)  \\
\end{bmatrix}
\left [\begin {array}{c}
u'  \\
v'  \\
w'
\end {array}\right ])^T,
\]
for $u,v,w,u',v',w'\in D$, i.e.
$$(u,v,w)(u',v',w')=(uu'+d\widetilde{\tau}(w)v'+d\widetilde{\tau}^2(v)w',
                     vu'+\widetilde{\tau}(u)v'+d\widetilde{\tau}^2(w)w',
                     wu'+\widetilde{\tau}(v)v'+\widetilde{\tau}^2(u)w').$$
 The multiplication in ${\rm It}_M^3(D,\tau,d)$ is given by
\[(u,v,w)(u',v',w')
=(
\begin{bmatrix}
u &\widetilde{\tau}(w) d & \widetilde{\tau}^2(v)d  \\
v & \widetilde{\tau}(u) & \widetilde{\tau}^2(w) d \\
w & \widetilde{\tau}(v) & \widetilde{\tau}^2(u)  \\
\end{bmatrix}
\left [\begin {array}{c}
u'  \\
v'  \\
w'
\end {array}\right ])^T,
\]
for $u,v,w,u',v',w'\in D$, hence
$$(u,v,w)(u',v',w')=(uu'+\widetilde{\tau}(w)dv'+\widetilde{\tau}^2(v)dw',
                     vu'+\widetilde{\tau}(u)v'+\widetilde{\tau}^2(w)dw',
                     wu'+\widetilde{\tau}(v)v'+\widetilde{\tau}^2(u)w').$$
\end{example}

 If $\{1,e,\ldots, e^{m-1}\}$ is the standard basis for $D$,
 then
\[\{1, e, \ldots, e^{m-1}, f, fe, \ldots, f^{n-1}e^{m-1}\}\]
is a basis for the right $K$-vector space $A$.
Writing elements in
$A$ as column vectors of length $mn$ with entries in $K$, we obtain
\[xy = \lambda(M(x)) y,\]
where
\begin{equation} \label{equ:main}
\lambda(M(x)) = \left[ \begin{array}{cccc}
\lambda(x_0) & \lambda(d) \tau(\lambda(x_{n-1}))&  \cdots & \lambda(d) \tau^{n-1}(\lambda(x_1)) \\
\lambda(x_1) & \tau(\lambda(x_0)) & \cdots & \lambda(d) \tau^{n-1}(\lambda(x_{2})) \\
\vdots & \vdots  & \ddots & \vdots \\
\lambda(x_{n-1}) & \tau(\lambda(x_{n-2})) & \cdots & \tau^{n-1}(\lambda(x_0)) \end{array} \right]
\end{equation}
for $A={\rm It}^n(D, \tau, d)$, and
\begin{equation}
\lambda(M(x)) = \left[ \begin{array}{cccc}
\lambda(x_0) &  \tau(\lambda(x_{n-1})) \lambda(d)&  \cdots &  \tau^{n-1}(\lambda(x_1))\lambda(d) \\
\lambda(x_1) & \tau(\lambda(x_0)) & \cdots &       \tau^{n-1}(\lambda(x_{2})) \lambda(d) \\
\vdots & \vdots  & \ddots & \vdots \\
\lambda(x_{n-1}) & \tau(\lambda(x_{n-2})) & \cdots & \tau^{n-1}(\lambda(x_0)) \end{array} \right]
\end{equation}
for $A={\rm It}_M^n(D, \tau, d)$,
 is the $mn \times mn$ matrix obtained by taking the left regular representation of each entry in the
 matrix $M(x)$.
The
matrix $\lambda(M(x))$ represents the left multiplication by the element $x$ in $A$.

\begin{remark} For all $X=\lambda(M(x)) =\lambda(x)\in
\lambda(A)\subset {\rm Mat}_{nm}(K),$ we have ${\rm det}\, X\in F$.
This is proved in  \cite{SP14} for $It^n(D, \tau, d)$. For ${\rm It}_M^n(D, \tau, d)$,
 the proof is  analogous.
(For $n=2$  this is \cite[Theorem 19]{P13.2}.)
\end{remark}

\begin{theorem} \label{thm:main}
(i) Let $x \in A$ be nonzero. If $x$ is not a left zero divisor in $A$, then ${\rm det}\,\lambda(M(x))\not=0$.
\\ (ii) $A$ is division if and only if $\lambda(M(x))$ is invertible for every nonzero $x \in A$.
\end{theorem}

\begin{proof}
(i) Suppose $\lambda(M(x))$ is a singular matrix. Then the system of $mn$ linear equations
\[\lambda(M(x)) (y_0, \ldots, y_{mn-1}) = 0\]
has a non-trivial solution $(y_0, \ldots, y_{mn-1}) \in K^{mn}$ which contradicts the assumption that $x$ is not
a left zero divisor in $A$.
\\ (ii) It remains to show that $\lambda(M(x))$ is invertible for every nonzero $x \in A$ implies that $A$ is division:
for all $x\not=0$, $y\not=0$ we have that $xy = \lambda(M(x)) y=0$ implies that $y=\lambda(M(x))^{-1}0=0$, a contradiction.
\end{proof}

The following  result concerning left zero divisors is proved analogously to
 \cite{R13}, Appendix A and requires Lemma \ref{nuclem}:

\begin{theorem}\label{le:mainthm}
If $d \neq z \tl(z) \tl^2(z) \ldots \tl^{n-1}(z)$ for all $z \in D$, then no element
$x=x_0 + fx_1\in A$ is a left zero divisor.
\end{theorem}

\subsection{}

In this section,  $A={\rm It}^n(D, \tau, d)$. We assume $d \in F^\times$, unless explicitly stated otherwise.

\begin{lemma}\label{nuclem}
 (i) If $d\not\in F_0$ then $D={\rm Nuc}_m(A)={\rm Nuc}_l(A)$.
\\ (ii) Let $F'$ and $L'$ be fields and let $K'$ be a cyclic extension of both $F'$ and $L'$ such that
 $Gal(K'/F) = \langle \sigma' \rangle$ and $[K:F] = m'$,
 $Gal(K'/L') = \langle \tau' \rangle$ and $[K:L] = n'$,
 $\sigma'$ and $\tau'$ commute.
Assume $F_0=F'\cap L'$ and $d'\in F'^\times$.
Let $D'=(K'/F', \sigma', c')$ be a cyclic division algebra over $F'$ of degree $m'$, $c'\in F_0$.
 If ${\rm It}^n(D, \tau, d)\cong {\rm It}^{n'}(D', \tau', d')$  then $D\cong D'$ and thus also $F\cong F'$, $m=m'$
and $n=n'$.
\end{lemma}

\begin{proof}
(i) follows from Theorem \ref{thm:skew} \cite[(2)]{P66}.
 \\
 (ii) follows from (i), since every isomorphism preserves the middle nucleus.
\end{proof}

\begin{theorem} \label{thm:skew}
 ${\rm It}^n(D,\tau,d)$
is a division algebra if and only if the polynomial $$f(t)=t^n-d$$ is irreducible in the twisted polynomial ring
$D[t;\widetilde{\tau}^{-1}]$.
\end{theorem}

\begin{proof}
 Let  $R=D[t;\widetilde{\tau}^{-1}]$  as defined in \cite{J96} and $f(t)=t^n-d\in R$.
 Let ${\rm mod}_r f$ denote the remainder of right division by $f$ in $R$.
 Then the vector space $V=\{g\in D[t;\widetilde{\tau}^{-1}]\,|\, {\rm deg}(g)<n\}$ together with the multiplication
 $$g\circ h=gh \,\,{\rm mod}_r f $$
 becomes a nonassociative algebra denoted $S_f=(V,\circ)$ over $F_0$
 \cite{P66}.  A straighforward calculation shows that
${\rm It}^n(D,\tau,d)= S_f$
\cite{P14}. By  \cite[p.~13-08 (9)]{P66}, ${\rm It}^n(D,\tau,d)= S_f$ is division if $f$ is irreducible.
Conversely, if $f=f_1f_2$ is reducible then $f_1$ and $f_2$ yield zero divisors in ${\rm It}^n(D,\tau,d)= S_f$.
\end{proof}

Theorem \ref{thm:skew} together with the results in \cite{P66} and \cite{B14}, cf. \cite{P14}, imply:

\begin{theorem}\label{mainthm}
(i) Suppose that $n$ is prime and in case $n\not=2,3$, additionally that $F_0$ contains a primitive $n$th root of unity.
 Then ${\rm It}^n(D,\tau,d)$ is a division algebra if and only if
$$d\not=z\widetilde{\tau}(z)\tl^2(z)\cdots \widetilde{\tau}^{n-1}(z)$$
for all $z\in D$.
\\ (ii)  (cf. \cite{B14}) ${\rm It}^4(D,\tau,d)$ is a division algebra
if and only if
$$d\not=z\widetilde{\tau}(z)\widetilde{\tau}^2(z)\widetilde{\tau}^3(z)$$
and
$$ \widetilde{\tau}^2(z_1)\widetilde{\tau}^3(z_1)z_1+\widetilde{\tau}^2(z_0)z_1+\widetilde{\tau}^2(z_1)
\widetilde{\tau}^3(z_0)\not=0 \text{ or } \widetilde{\tau}^2(z_0)z_0+\widetilde{\tau}^2(z_1)\widetilde{\tau}^3(z_0)z_0\not=d$$
 for all $z,z_0,z_1\in D$.
\end{theorem}

From Theorem \ref{thm:skew} we obtain:

\begin{theorem} (equivalent to \cite[Theorems 20, 21]{P14})
 Let $F_0$ be of characteristic not 2 and $d\in F\setminus F_0$.
\\ (i) If $D=(a,c)_{F_0}\otimes_{F_0}F$ is a division algebra over $F$,
then
${\rm It}^2(D,\tau,d)$ is a division algebra.
\\ (ii) Let $F=F_0(\sqrt{b})$.
 Let $D_0=(L/F_0,\sigma,c)$ be a cyclic algebra of degree $3$
  such that $D=D_0\otimes_{F_0}F$ is a division algebra over $F$. If $d=d_0+\sqrt{b}d_1\in F\setminus F_0$
with $d_0,d_1\in F_0$, such that
 $3d_0^2+bd_1^2\not=0$,  then
 ${\rm It}^2(D,\tau,d)$ is a division algebra.
 \\ In particular, if $F_0=\mathbb{Q}$ and $b>0$, or if $b<0$ and $-\frac{b}{3}\not\in\mathbb{Q}^{\times 2}$
  then ${\rm It}^2(D,\tau,d)$ is a division algebra.
 \end{theorem}

\begin{proposition} \label{prop:nec} \cite{P14}
Suppose that $n$ is prime and in case $n\not=2,3$, additionally that $F_0$ contains a primitive $n$th root of unity.
  If $\tau(d^m)\not=d^m$, then ${\rm It}^n(D,\tau,d)$ is a division algebra.
\end{proposition}

Note that this generalizes \cite[Proposition 13]{MO13}.

\begin{proposition}  \label{prop:6}
Suppose that $n$ is prime and in case $n\not=2,3$, additionally that $F_0$ contains a primitive $n$th root of unity.
 If $d\in F$ such that $d^m\not\in N_{D/F_0}(D^\times)$, then ${\rm It}^n(D,\tau,d)$
 is a division algebra. In particular, for all $d\in F\setminus F_0$ with $d^m\not\in F_0$, ${\rm It}^n(D,\tau,d)$
 is a division algebra.
\end{proposition}

\begin{proof}
Since $c\in F_0\subset {\rm Fix}(\tau)=L$ we have
$N_{D/F}(\widetilde{\tau}(x))=\tau(N_{D/F}(x))$ for all $x\in D$ by \cite[Proposition 4]{P13.2}.
Assume that $d=z\tl(z)\cdots\tl^{n-1}(z)$, then
$$N_{D/F}(d)=N_{D/F}(z)N_{D/F}(\tl(z))\cdots N_{D/F}(\tl^{n-1}(z))=N_{D/F}(z) \tau(N_{D/F}(z))\cdots \tau^{n-1}
(N_{D/F}(z)).$$
Put $a=N_{D/F}(z)$, note that $a \tau(a)\cdots\tau^{n-1}(a)=N_{F/F_0}(a)=N_{F/F_0}(N_{D/F}(z))=N_{D/F_0}(z)\in F_0$,
 and use that $N_{D/F}(d)=d^m$ for $d\in F$.
\end{proof}

\section{General iteration process III: Natarajan and Rajan's  algebras}

\subsection{}
We use the same setup and notation as in Section 3 and now formally define the algebra behind the codes in~\cite{R13}.

\begin{definition} (B. S. Rajan and L. P. Natarajan)
Pick $d \in D^\times$.
 Define a multiplication on the right $D$-module
\[ D \oplus fD \oplus f^2D \oplus \cdots \oplus f^{n-1}D,\]
via the rules
\[
 (f^i x)(f^j y) =
  \begin{cases}
   f^{i+j} \tl^j(x)y & \text{if } i+j < n \\
   f^{(i+j)-n} \tl^j(x)yd  & \text{if } i+j \geq n
  \end{cases}
\]
for all $x, y \in D$, $i,j<n$, and call the resulting algebra ${\rm It}_R^n(D, \tau, d) $.
\end{definition}

 ${\rm It}_R^n(D,\tau,d)$ is an algebra over $F_0$ of dimension $nm^2[F:F_0]$ with unit element
 $1\in D$, contains $D$ as a subalgebra, and $f^{n-1}f=d=ff^{n-1}$.  For $d\in F^\times$,
 $${\rm It}_R^n(D,\tau,d)={\rm It}^n(D,\tau,d)={\rm It}_M^n(D,\tau,d).$$
If $d\in F_0$ and $F\not=L$ then ${\rm It}_R^n(D,\tau,d)={\rm It}^n(D,\tau,d)={\rm It}_M^n(D,\tau,d)$ is an
associative $F_0$-algebra, cf. \cite[Remark 1]{R13}. ${\rm It}_R^2(D,\tau,d)={\rm It}_R(D,\tau,d)$ is an iterated algebra.

\begin{lemma}\label{nuclem2}
Let $A={\rm It}_R^n(D,\tau,d)$.
\\ (i) If $d\not\in F_0$ then $D={\rm Nuc}_m(A)={\rm Nuc}_l(A)$.
\\ (ii) Let $F'$ and $L'$ be fields and let $K'$ be a cyclic extension of both $F'$ and $L'$ such that
 $Gal(K'/F) = \langle \sigma' \rangle$ and $[K:F] = m'$,
 $Gal(K'/L') = \langle \tau' \rangle$ and $[K:L] = n'$,
 $\sigma'$ and $\tau'$ commute.
Assume $F_0=F'\cap L'$.
Let $D'=(K'/F', \sigma', c')$ be a cyclic division algebra over $F'$ of degree $m'$, $c'\in F_0$.
 If ${\rm It}_R^n(D, \tau, d)\cong {\rm It}_R^{n'}(D', \tau', d')$  then $D\cong D'$ and thus also $F\cong F'$, $m=m'$
and $n=n'$.
\\  (iii) If $d\in K^\times$, then the (associative or nonassociative) cyclic algebra $(K/L,\tau,d)$ of degree $n$, viewed as algebra over $F_0$,
 is a subalgebra of $A$.
\\ (iv) For $n>3$, $n$ even,  ${\rm It}_R(D,\tau,d)$ is isomorphic to a proper subalgebra of $It_R^n(D, \tau, d)$.
\\ (v) $A\otimes_F K\cong{\rm Mat}_m(K)\oplus f{\rm Mat}_m(K)\oplus\dots\oplus f^{n-1}{\rm Mat}_m(K)$
 contains the $F_0$-algebra
${\rm Mat}_m(K)$ as subalgebra and has zero divisors.
\end{lemma}

The proofs of (i) and (ii) are analogous to the one of Lemma \ref{nuclem}, the ones of
(iii), (iv), (v) to the ones in Lemma \ref{lem:lem3}.

\begin{theorem} (i) \label{thm:skew2}
 ${\rm It}_R^n(D,\tau,d)$
is a division algebra if and only if the polynomial $$f(t)=t^n-d$$ is irreducible in the twisted polynomial ring
$D[t;\widetilde{\tau}^{-1}]$.
\\ (ii) Suppose that $n$ is prime and in case $n\not=2,3$, additionally that $F_0$ contains a primitive $n$th root of unity.
 Then ${\rm It}_R^n(D,\tau,d)$ is a division algebra if and only if
$$d\not=z\widetilde{\tau}(z)\tl^2(z)\cdots \widetilde{\tau}^{n-1}(z) $$
 for all $z\in D$.
 \\ (iii)  (cf. \cite{B14}) ${\rm It}_R^4(D,\tau,d)$ is a division algebra
if and only if
$$d\not=z\widetilde{\tau}(z)\widetilde{\tau}^2(z)\widetilde{\tau}^3(z)$$
and
$$ \widetilde{\tau}^2(z_1)\widetilde{\tau}^3(z_1)z_1+\widetilde{\tau}^2(z_0)z_1+\widetilde{\tau}^2(z_1)
\widetilde{\tau}^3(z_0)\not=0 \text{ or } \widetilde{\tau}^2(z_0)z_0+\widetilde{\tau}^2(z_1)\widetilde{\tau}^3(z_0)z_0\not=d$$
 for all $z,z_0,z_1\in D$.
 \\ (iv) \label{prop:nec1}
 Suppose that $n$ is prime and  in case $n\not=2,3$, additionally that $F_0$ contains a primitive $n$th root of unity. Let $d\in K\setminus L$.
If $\tau(d^m)\not=d^m$, then ${\rm It}_R^n(D,\tau,d)$ is a division algebra.
\end{theorem}

\begin{proof}
(i) Let  $R=D[t;\widetilde{\tau}^{-1}]$ and $f(t)=t^n-d\in R$. Since
${\rm It}_R^n(D,\tau,d)= S_f$
\cite{P14},  the assertion now follows as in the proof of Theorem \ref{thm:skew}.
\\ (ii), (iii) and (iv) follow from (i) together with the improvements of the conditions \cite[(18)(19)]{P66} given in
\cite{B14}, cf. \cite{P14}.
\end{proof}

\begin{proposition}  \label{prop:6.2}
 Suppose that $n$ is prime and in case $n\not=2,3$, additionally  that $F_0$ contains a primitive $n$th root of unity.
\\ (i)  If $N_{D/F}(d)\not\in N_{D/F_0}(D^\times)$
then ${\rm It}_R^n(D,\tau,d)$ is a division algebra.
\\ (ii)  ${\rm It}_R^n(D,\tau,d)$ is division for all $d\in D$ such that
$N_{D/F}(d)\not\in F_0$.

\end{proposition}

\begin{proof}
(i) Since $c\in F_0\subset {\rm Fix}(\tau)=L$ we have
$N_{D/F}(\widetilde{\tau}(x))=\tau(N_{D/F}(x))$ for all $x\in D$ by \cite[Proposition 4]{P13.2}.
Assume that $d=z\tl(z)\cdots\tl^{n-1}(z)$ for some $z\in D$, then
$$N_{D/F}(d)=N_{D/F}(z)N_{D/F}(\tl(z))\cdots N_{D/F}(\tl^{n-1}(z))=N_{D/F}(z) \tau(N_{D/F}(z))\cdots \tau^{n-1}
(N_{D/F}(z)).$$
Put $a=N_{D/F}(z)$ and note that $a \tau(a)\cdots\tau^{n-1}(a)=N_{F/F_0}(a)=N_{F/F_0}(N_{D/F}(z))=N_{D/F_0}(z)\in F_0$.
\\ (ii) follows from (i).
\end{proof}

\begin{theorem} \label{thm:degree3}
Let $F_0$ have characteristic not 2.
Let $D=(e,c)_F$, $c\in F_0$, be a quaternion division algebra over $F$.
\\ (i) If $[K:L]=3$ and $d\in L\setminus F_0$ such that $d\not\in N_{K/L}(K^\times)$,
then
${\rm It}_R^3(D,\tau,d)$ is a division algebra.
\\ (ii) If $[K:L]=4$ and $d\in L\setminus F_0$ such that $d^s\not\in N_{K/L}(K^\times)$ for $t=1,2,3$, then
$d\not=z\widetilde{\tau}(z)\widetilde{\tau}^2(z)\widetilde{\tau}^3(z)$ for all $z\in D$.
\end{theorem}

\begin{proof}
(i) By Theorem \ref{thm:skew2} (ii), ${\rm It}_R^3(D,\tau,d)$ is a division algebra
if and only if
 $d\not=z\widetilde{\tau}(z)\widetilde{\tau}^2(z)$
 for all $z\in D$. Suppose that
 \begin{equation}\label{eqn:3}
 d=z\widetilde{\tau}(z)\widetilde{\tau}^2(z)
 \end{equation}
  for some $z=a+jb\in D$, $a,b\in K$, then
 $a\not=0$ and $b\not=0$: suppose $a=0$, then $jbj\widetilde{\tau}(b)j^2\widetilde{\tau}^2(b)\in Kj$ contradicts that
 $d\in L^\times$; suppose $b=0$, then $d=a\tau(a)\tau^2(a)=N_{K/L}(a)$ contradicts that $d\not\in N_{K/L}(K^\times)$.
Equation (\ref{eqn:3}) implies that $\widetilde{\tau}^2(d)=\widetilde{\tau}^2(z) z\widetilde{\tau}(z)$, since $d\in L$
therefore
$$z\widetilde{\tau}(z)\widetilde{\tau}^2(z)=\widetilde{\tau}^2(z)z\widetilde{\tau}(z).$$
Thus for $D=(e,c)_F$, $c\in F_0$,
$$z\widetilde{\tau}(z)=x+j\sigma(y)$$
with $x=a\tau(a)+c\sigma(b)\tau(b)$, $\sigma(y)=b\tau(a)+\sigma(a)\tau(b)$. From
$(x+j\sigma(y))(\tau^2(a)+j\tau^2(b))=(\tau^2(a)+j\tau^2(b))(x+j\sigma(y))$ it follows that
$$\sigma(x)\tau^2(b)+\sigma(y)\tau^2(a)=\tau^2(b)x+\sigma(\tau^2(a))\sigma(y).$$
Equation (\ref{eqn:3}) yields
 \begin{equation}\label{eqn:4}
d=x\tau^2(a)+cy\tau^2(b)
 \end{equation}
and
 \begin{equation}\label{eqn:5}
0=\sigma(x)\tau^2(b)+\sigma(y)\tau^2(a).
 \end{equation}
Now $x\not=0$ (or else we get a contradiction), so Equations (\ref{eqn:4}) and (\ref{eqn:5}) together with Equation (\ref{eqn:3}) imply that
$$\frac{\sigma(y)}{\tau^2(b)}=-\frac{x}{\sigma(\tau^2(a))}=\frac{\sigma(x)}{\tau^2(a)}$$
and
$$\frac{-\sigma(x)}{\tau^2(a)}=\frac{\sigma(y)}{\tau^2(b)},$$
so that
$$\frac{\sigma(x)}{\tau^2(a)}\in {\rm Fix}(\sigma)=F.$$
 Use Equation (\ref{eqn:5}) in Equation (\ref{eqn:4}) to obtain
$$\frac{\tau^2(a)}{\sigma(x)}(x\sigma(x)-cy\sigma(y))=d.$$
Since $x\sigma(x)-cy\sigma(y)\in F$, the left-hand side lies in $F$, contradicting the choice of $d\in L\setminus F$.
   Thus $d\not=z\widetilde{\tau}(z)\widetilde{\tau}(z)^2$.
   \\ (ii) The proof is a straightforward calculation analogous to (i) or the proof of \cite[Proposition 5]{R13}.
\end{proof}

\subsection{} $A={\rm It}_R^n(D, \tau, d)$ is  a right $K$-vector space of dimension $mn$.
 By choosing $d\in L^\times$ from now on, we achieve that left multiplication
 $L_x$ is a $K$-endomorphism and can be represented
by a matrix with entries in $K$.

For $d\in L^\times$,  the algebras $A={\rm It}_R^n(D,\tau,d)$ are behind the codes
defined by  Srinath and  Rajan \cite{R13}, even though they are not explicitly defined there as such.
In the setup of \cite{R13}, it is assumed that $d\in L\setminus F$ and that $L\not=F$.
 We do not assume that $L\not=F$ for now.

\begin{example}
Let $F_0$ have characteristic not 2 and  $D = (K/F, \sigma, c)=K \oplus eK$ be a quaternion division algebra over $F$ with
multiplication
\begin{equation}\label{eqn:1}
(x_0 + ex_1)(u_0 + eu_1) = \big(x_0u_0 + c \sigma(x_1)u_1 \big) + e\big(x_1u_0 + \sigma(x_0)u_1 \big),
\end{equation}
for $x_i, u_i \in K$. Let $K/L$ be a quadratic field extension with non-trivial automorphism $\tau$, $d \in K^\times$. The
iterated algebra
\[{\rm It}_R(D, \tau, d) = D \oplus fD  = K \oplus eK \oplus fK \oplus feK,\]
has multiplication
\[(x+fy)(u+fv) = \big( xu+\tl(y)vd \big) + f \big( yu + \tl(x)v \big),\]
where
$x = x_0 + ex_1$, $y = y_0 + ey_1$, $u = u_0 + eu_1,$ $v = v_0 + ev_1\in D, $ $x_i,y_i,u_i,v_i\in K$.
Here,
$$xu \text{ is given in equation } (\ref{eqn:1}),$$
\begin{align*}
\tl(y)vd = (\tl(y)v)d & = \Big(\big( \tau(y_0) v_0 + c \sigma \tau(y_1)v_1 \big) + e \big(\tau(y_1)v_0 + \sigma \tau(y_0)v_1 \big) \Big) (d+ e0)\\
&=\big( \tau(y_0) v_0 d+ c \sigma \tau(y_1)v_1 d\big) + e \big(\tau(y_1)v_0 d+ \sigma \tau(y_0)v_1d \big),
\end{align*}
\begin{align*}
yu &= \big( y_0 u_0 + c \sigma(y_1)u_1 \big) + e \big( y_1u_0 + \sigma(y_0)u_1 \big),\\
\tl(x)v &= \big( \tau(x_0)v_0 + c \sigma\tau(x_1)v_1 \big) + e \big( \tau(x_1)v_0 + \sigma\tau(x_0) v_1 \big).
\end{align*}
Thus we can write the multiplication in terms of the $K$-basis $\{1, e, f, fe\}$ as
\begin{align*}
(x+fy)(u+fv) &= \big(x_0u_0 + c \sigma(x_1)u_1 + \tau(y_0) v_0 d+ c \sigma \tau(y_1)v_1 d\big)\\
&+ e \big( x_1u_0 + \sigma(x_0)u_1 + \tau(y_1)v_0 d+ \sigma \tau(y_0)v_1d \big)\\
&+ f \big( y_0 u_0 + c \sigma(y_1)u_1 +\tau(x_0)v_0 + c \sigma\tau(x_1)v_1 \big)\\
&+ fe \big(  y_1u_0 + \sigma(y_0)u_1 +\tau(x_1)v_0 + \sigma\tau(x_0) v_1 \big).
\end{align*}
Since $d \in K$, it commutes with the elements $y_i$ and $v_i$ in the above expression.

Write $\Phi(x+fy)$ for the column vector with respect to the $K$-basis, i.e.,
\[\Phi(x+fy) = (x_0, x_1, y_0, y_1)^T \hspace{3mm} \text{ and } \hspace{3mm} \Phi(u+fv) = (u_0, u_1,v_0,v_1)^T,\]
then we can write the product as
\[\Phi((x+fy)(u+fv)) = \begin{bmatrix}
x_0u_0 + c \sigma(x_1)u_1 + d\tau(y_0) v_0 + dc \sigma \tau(y_1)v_1 \\
x_1u_0 + \sigma(x_0)u_1 + d\tau(y_1)v_0 + d\sigma \tau(y_0)v_1\\
y_0 u_0 + c \sigma(y_1)u_1 +\tau(x_0)v_0 + c \sigma\tau(x_1)v_1\\
y_1u_0 + \sigma(y_0)u_1 +\tau(x_1)v_0 + \sigma\tau(x_0) v_1
\end{bmatrix}=\]
\[\begin{bmatrix}
x_0 & c \sigma(x_1) & d\tau(y_0) & dc \sigma \tau(y_1)\\
x_1 & \sigma(x_0) & d\tau(y_1) & d\sigma \tau(y_0)\\
y_0 & c \sigma(y_1) & \tau(x_0) & c \sigma\tau(x_1)\\
y_1 & \sigma(y_0) & \tau(x_1) & \sigma\tau(x_0)
\end{bmatrix} \begin{bmatrix}
u_0\\ u_1 \\ v_0 \\ v_1
\end{bmatrix}.\]
The matrix on the left side is equal to
\[\begin{bmatrix}
\lambda(x) & d \lambda( \tl(y))\\
\lambda(y) & \lambda( \tl(x))
\end{bmatrix}. \]
Thus for $d\in L^\times$, left multiplication $L_x$ is a $K$-endomorphism and can be represented
by the above matrix with entries in $K$.
\end{example}

In the following, $A={\rm It}_R^n(D, \tau, d)$ and we assume $d\in L^\times$.
Any element $x \in A$ can be identified with a unique column vector
$\Phi(x) \in K^{mn }$
using the standard $K$-basis
$$\{1,e,\dots,e^{m-1},f,fe,\dots,fe^{m-1},\dots,f^{n-1},f^{n-1}e,\dots,f^{n-1}e^{m-1}\}.$$
For $x=x_0+fx_1+f^2x_2+\cdots+f^{n-1}x_{n-1}$, $x_0,\dots,x_{n-1} \in D$, define
\begin{equation} \label{equ:matrix_rep_A}
\Lambda(x) = \lambda(M(x)) = \begin{bmatrix}
             \lambda(x_0) & d \tau(\lambda(x_{n-1}))& d \tau^2(\lambda(x_{n-2})) & \cdots & d \tau^{n-1}(\lambda(x_1)) \\
             \lambda(x_1) & \tau(\lambda(x_0)) & d \tau^2(\lambda(x_{n-1})) & \cdots & d \tau^{n-1}(\lambda(x_{2})) \\
             \lambda(x_2) & \tau(\lambda(x_1)) & \tau^2(\lambda(x_0)) & \cdots & d \tau^{n-1}(\lambda(x_3))\\
             \vdots & \vdots & \vdots & \ddots & \vdots \\
             \lambda(x_{n-1}) & \tau(\lambda(x_{n-2})) & \tau^2(\lambda(x_{n-3})) & \cdots & \tau^{n-1}(\lambda(x_0))
             \end{bmatrix},
\end{equation}
where $\lambda(x_i)$, $x_i \in D$, is the $m \times m$ matrix with entries in $K$ representing
 left multiplication by $x_i$ in the cyclic division algebra $D$.

\begin{lemma} (B. S. Rajan and L. P. Natarajan) \label{lem:matrix_rep_A}
\\
(i) For any $x \in {\rm It}_R^n(D, \tau, d)$, $\Lambda(x)$
 is the matrix representing left multiplication by $x$ in $A$, i.e., $\Phi(xy)=\Lambda(x)\Phi(y)$ for every
$y \in A$. \\
(ii) $A$ is division if and only if $\Lambda(x)=\lambda(M(x))$ is invertible for every nonzero $x \in A$.\\
(iii) For every $x \in A$, $det(\lambda(M(x))) \in L$.
\end{lemma}

\begin{proof}
(i) For any $r \in D$ with $r = r_0 + er_1 + \dots +er^{m-1}$, where $r_0,\dots,r_{m-1} \in K$,
define $\phi(r)=\left[r_0, r_1,\dots,r_{m-1} \right]^T$. Let $x=\sum_{i=0}^{n-1}f^ix_i$ and
$y=\sum_{j=0}^{n-1}f^jy_j$ be elements of $A$, with $x_i, y_i \in D$. For any $x_i$ and $y_i$,
the multiplication in $D$ is given by $\phi(x_iy_i) = \lambda(x_i)\phi(y_i)$.
Moreover, since $d \in L$ we see that $\phi(d) = \left[d,0,\ldots, 0 \right]$, and therefore
\[\phi(y_i d) = \lambda(y_i)\phi(d) =  \phi(y_i)d = d \phi(y_i).\]
Now it is straightforward to see that the matrix multiplication $\Lambda(x) \Phi(y)$ does indeed represent the
multiplication in $A$.
\\
(ii)  $A$ is division if and only if $xy \neq 0$ for every nonzero $x,y \in A$~\cite{Sch2},
 i.e., if and only if $\Lambda(x)\Phi(y) \neq 0$, or equivalently, if and only if $\Lambda(x)$ is invertible for every
 nonzero $x \in A$.
\\
(iii) It enough to show that $\det(\Lambda(x))=\tau(\det(\Lambda(x))) = \det(\tau(\Lambda(x)))$, where
\begin{equation*}
\tau(\Lambda(x)) = \begin{bmatrix}
             \tau(\lambda(x_0)) & d \tau^2(\lambda(x_{n-1}))& \cdots & d \tau^{n-1}(\lambda(x_{2})) & d\lambda(x_1) \\
             \tau(\lambda(x_1)) & \tau^2(\lambda(x_0)) & \cdots & d \tau^{n-1}(\lambda(x_{3})) & d\lambda(x_{2}) \\
             \tau(\lambda(x_2)) & \tau^2(\lambda(x_1)) & \cdots & \tau^{n-1}(\lambda(x_4)) & d\lambda(x_3)\\
             \vdots & \vdots & \vdots & \ddots & \vdots \\
             \tau(\lambda(x_{n-1})) & \tau^2(\lambda(x_{n-2})) & \cdots & \tau^{n-1}(\lambda(x_{1})) & \lambda(x_0)
             \end{bmatrix}.
\end{equation*}
It follows that $\Lambda(x)=P\tau(\Lambda(x))P^{-1}$, with
\[P=\left[ \begin{array}{ccccc}
{0} & {0} & \cdots & {0} & d {I}_m\\
{I}_m & {0} & \cdots & {0} &{0} \\
{0} & {I}_m & \cdots & {0} & {0} \\
\vdots & \vdots & \ddots & \vdots & \vdots\\
{0} & {0} & \cdots & {I}_m & {0}
\end{array}\right] \textrm{ and }
P^{-1}=\left[ \begin{array}{ccccc}
{0} & {I}_m & {0} & \cdots & {0}\\
{0} & {0} & {I}_m & \cdots &{0} \\
\vdots & \vdots & \ddots & \vdots & \vdots\\
{0} & {0} & {0} & \cdots & {I}_m\\
d^{-1}{I}_m & {0} & {0} & \cdots & {0}
\end{array}\right],\]
where $I_m$ is the $m \times m$ identity matrix and $0$ is the $m \times m$ zero matrix which proves the assertion.
\end{proof}

\section{How to design fast-decodable Space-Time Block Codes using ${\rm It}_R^n(D, \tau, d)$}

\subsection{} \label{5.1}
To construct fully diverse space-time block codes for $mn$ transmit antennas using ${\rm It}_R^n(D, \tau, d)$ (or
${\rm It}^n(D, \tau, d)$ in the next Section), let
$L$ be either $\mathbb{Q}(i)$ or $\mathbb{Q}(\omega)$, $\omega=e^{2\pi i /3}$,
and $D = (K/F, \sigma, c)$ a cyclic division algebra of degree $m$ over a number field $F\not=L$,
$c \in F\cap L$, and
where $K$ is a cyclic extension of $L$ of degree $n$ with Galois group generated by $\tau$.
We assume that $\sigma$ and $\tau$ commute. For $x\in D$, let $\lambda(x)$ be the $m \times m$ matrix with entries in $K$
given by the left regular representation in $D$.

Each entry of $\lambda(x)$ can be viewed as a linear
combination of $n$ independent elements of $L$. As such we express each entry of these as a linear combination
of some chosen $L$-basis $\{\theta_1, \theta_2, \ldots, \theta_n \mid \theta_i \in \mathcal{O}_K\}$ over $\mathcal{O}_L$.
Thus an entry $\lambda(x)$ has the form
\begin{equation}\label{lambdax}
\lambda(x) = \left[ \begin{array}{cccc}
\sum_{i=1}^n s_i \theta_i & c \sigma(\sum_{i=1}^n s_{i+nm-n} \theta_i)  & \dots  & c \sigma^{m-1}(\sum_{i=1}^n s_{i+n} \theta_i) \\
\sum_{i=1}^n s_{i+n} \theta_i & \sigma(\sum_{i=1}^n s_i \theta_i)  & \dots  & c \sigma^{m-1}(\sum_{i=1}^n s_{i+2n} \theta_i )\\
\vdots  &   \vdots      &   \ddots   &  \vdots          \\
\sum_{i=1}^n s_{i+nm-n} \theta_i & \sigma(\sum_{i=1}^n s_{i+nm-2n} \theta_i) & \dots  &  \sigma^{m-1}(\sum_{i=1}^n s_i \theta_i) \end{array}
 \right].
\end{equation}
 The elements $s_i, 1 \leq i \leq mn$, are the complex information symbols with values from QAM ($\mathbb{Z}(i)$) or HEX
  ($\mathbb{Z}(\omega)$) constellations.

\subsection{} \label{5.2}
We assume that $f(t)=t^n-d\in D[t;\tl^{-1}]$, $d\in L^\times$, is irreducible.
Then $A={\rm It}_R^n(D, \tau, d)$ is division and each codeword in $\mathcal{C}$ is a matrix of the form given in
(\ref{equ:matrix_rep_A}) and these are
 invertible $mn\times mn$ matrices with entries in $K$.

Contrary to \cite{R13}, we are interested in high data rate, i.e. we use the $mn^2$ degrees of freedom of the channel to transmit
 $mn^2$ complex information symbols per codeword.
If $mn$ channels are
used the space-time block code $\mathcal{C}$ consisting of matrices $S$ of the form (\ref{equ:matrix_rep_A}) with entries as in (\ref{lambdax})
 has a rate of $n$ complex symbols per channel use, which is maximal for $n$ receive antennas.

\begin{proposition} \label{prop:groupdeco}
If the subset of codewords in $\mathcal{C}$ made up of the diagonal block matrix
$$S(\lambda(x_0))=diag [ \lambda(x_0),  \tau(\lambda(x_0))\dots, \tau^{n-1}(\lambda(x_0))]$$
 is $l$-group decodable, then $\mathcal{C}$ has
ML-decoding complexity  $\mathcal{O}(M^{mn^2-mn(l-1)/l)})$ and is fast-decodable.
\end{proposition}

\begin{proof}
To analyze ML-decoding complexity, we have to minimize the ML-complexity metric
$$||Y-\sqrt{\rho}HS||^2$$
over all codewords $S\in \mathcal{C}$.
Every $S\in \mathcal{C}$ can be written as
$$S=S(\lambda(x_0))+S(\lambda(x_1))+\dots+S(\lambda(x_{n-1}))$$
 with
$S(\lambda(x_0))=diag[\lambda(x_0),\tau(\lambda(x_0)),\dots,\tau^{n-1}(\lambda(x_0))]$ and $S(\lambda(x_j))$ being the matrix
obtained by
putting $\lambda(x_j)=0$, for all $j\not=i$  in (\ref{equ:matrix_rep_A}). Each $S(\lambda(x_i))$ contains $nm$ complex information symbols.
Since $S(\lambda(x_0))$ is $l$-group decodable by assumption, we need $\mathcal{O}(M^{nm/l})$
computations to compute ${\rm min}_{S(\lambda(x_0)))}\{||Y-\sqrt{\rho}HS||^2\}$.
So the $ML$-decoding complexity of $\mathcal{C}$ is $\mathcal{O}(M^{(n-1)(nm)+nm/l})=\mathcal{O}(M^{mn^2-mn(l-1)/l)})$
\end{proof}

\begin{corollary} \label{cor:ML}
If $D={\rm Cay}(K/F,-1)$ is a subalgebra of Hamilton's quaternion algebra $\mathbb{H}$ and
$d\in L\setminus F$, then the corresponding code $\mathcal{C}$ in (\ref{equ:matrix_rep_A})  has
decoding complexity
$$\mathcal{O}(M^{2n^2-3n/2})$$
 if the $s_i$ take values from $M$-QAM and decoding complexity
$$\mathcal{O}(M^{2n^2-n})$$
if the $s_i$ take values from $M$-HEX.
\end{corollary}

\begin{proof}
If $D={\rm Cay}(K/F,-1)$ is a quaternion division algebra which is a subalgebra of $\mathbb{H}$,
$\sigma$ commutes with complex conjugation, and a code
consisting of the block diagonal matrices $S(\lambda(x_0))$ above
 with entries as in (\ref{lambdax})
is four-group decodable if we the $s_i$ take values from $M$-QAM
and two group-decodable if  the $s_i$ take values from $M$-HEX. Consequently, $\mathcal{C}$ has
decoding complexity $\mathcal{O}(M^{(n-1)(2n)+n/2})=\mathcal{O}(M^{2n^2-3n/2})$ if the $s_i$ take values from $M$-QAM and
 decoding
complexity $\mathcal{O}(M^{(n-1)(2n)+n})=\mathcal{O}(M^{2n^2-n})$
if the $s_i$ take values from $M$-HEX \cite[Proposition 7 ff.]{R13}.
\end{proof}

\subsection{Specific code examples} \label{sec:Specific_Code_Examples}

 The Alamouti code has the best coding gain among known $2\times 1$
 codes of rate one, hence in our examples we will use $D=(-1,-1)_F$.

Our three code examples have high data rate and use the same algebras and automorphisms as the examples of \cite{R13}:
 Since the Alamouti code has the lowest ML-decoding complexity among the STBCs obtained from associative
division algebras, the choice of $D$ as a a subalgebra of Hamilton's quaternions in each example guarantees best possible
fast decodability.
 The choice of  $L$ and $K$ in \cite{R13} seems optimal to us as well since the extensions
are related to the corresponding perfect STBCs in the respective dimensions.
We start  building two codes using $A={\rm It}_R^n(D,\tau,d)$.

\subsection{Example of $6 \times 3$ MIMO System} \label{subsec:I}
Take the setup of~\cite[Section~IV.C.]{R13}.
Let $\omega=\frac{-1+\sqrt{3}i}{2}$ be a primitive third root of unity, $\theta = \zeta_7 + \zeta_7^{-1} = 2
\cos(\frac{2 \pi}{7})$, where $\zeta_7$ is a primitive $7^{th}$ root of
unity and let $F = \mathbb{Q}(\theta)$. Let $K = F(\omega) =\mathbb{Q}(\omega, \theta)$ and take $D = (K/F, \sigma, -1)$
as the quaternion division algebra. Note that $\sigma:i \mapsto -i$ and
therefore $\sigma(\omega) = \omega^2$.
 Let $L =\mathbb{Q}(\omega)$, so that $K/L$ is a cubic cyclic field extension
whose Galois group is generated by the automorphism $\tau: \zeta_7 +\zeta_7^{-1} \mapsto \zeta_7^2 + \zeta_7^{-2}$.
 We do not need to restrict our considerations to a sparse code as done in \cite{R13}  in order to get a fully
diverse code:

Since $\omega\not\in N_{K/L}(K^\times)$,
 $ {\rm It}_R^3(D, \tau, \omega)$ is a division algebra by Theorem \ref{thm:degree3}.
Hence the code consisting of all matrices of the
form
\[\left[ \begin{array}{ccc}
\lambda(x) & \omega\lambda(\tl(z)) & \omega \lambda(\tl^2(y))\\
\lambda(y) & \lambda(\tl(x)) & \omega\lambda(\tl^2(z))\\
\lambda(z) & \lambda(\tl(y)) & \lambda(\tl^2(x))
\end{array}  \right],\]
with $x,y,z$ not all zero, is fully diverse.
 Write $x = x_0 + ex_1$, $y = y_0 + ey_1$, $z = z_0 + ez_1$, where $x_i, y_i,z_i \in K$,
then its $6\times 6$ matrix is given by
\[S=\left[ \begin{array}{cccccc}
x_0 & -\sigma(x_1) &  \omega\tl(z_0) & -\omega\tl \sigma(z_1) & \omega \tl^2(y_0) & -\omega \tl^2 \sigma(y_1)\\
x_1 & \sigma(x_0) &  \omega\tl(z_1) &  \omega\tl \sigma(z_0) & \omega \tl^2 \sigma(y_1) & \omega \tl^2\sigma(y_0)\\
y_0 & -\sigma(y_1) & \tl(x_0) & -\tl \sigma(x_1) & \omega \tl^2(z_0) & -\omega \tl^2\sigma(z_0) \\
y_1 & \sigma(y_0) & \tl(x_1) & \tl\sigma(x_0) & \omega \tl^2(z_1) & \omega \tl^2\sigma(z_1) \\
z_0 & \sigma(z_0) & \tl(y_0) & -\tl\sigma(y_1) & \tl^2(x_0) & -\tl^2\sigma(x_1) \\
z_1 & -\sigma(z_1) & \tl(y_1) & \tl \sigma(y_0) & \tl^2(x_1) & \tl^2 \sigma(x_0)
\end{array} \right].\]
 With the encoding from \ref{5.1}, we  encode 18
complex information symbols with each codeword $S$. The code has rate $3$ for 6 transmit and 3 receive antennas, i.e.
maximal rate.

We use $M$-HEX complex constellations  and the notation from \ref{5.1} (i.e., $s_j\in\mathbb{Z}[\omega]$):
 choose $\{\theta_1,\theta_2,\theta_3\}$ to be a basis of the principal ideal in $\mathcal{O}_K$ generated
by $\theta_1$ with $\theta_1=1+\omega+\theta$, $\theta_2=-1-2\omega+\omega\theta^2$, $\theta_3=(-1-2\omega)+
(1+\omega)\theta+(1+\omega)\theta^2$. Since all entries of the code matrix $S$ lie in $\mathcal{O}_K$, here
$\det(S)\in\mathcal{O}_L=\mathbb{Z}[\omega]$
by Lemma \ref{lem:matrix_rep_A} (iii). Then
the determinant of any nonzero codeword $S$ is an element in $\mathbb{Z}[\omega]$ and, being fully diverse, the code has NVD
 which means the code is DMT-optimal \cite{SR3}.
 Its minimum determinant (of the unnormalized code) is thus at least 1.
By a similar  argument as given in \cite[C.]{R13}, using a normalization factor of $1/\sqrt{28E}$, the normalized minimum
 determinant is
 $$49(\frac{2}{\sqrt{28E}})^{18}=1/7^7 E^9.$$
Each codeword
$S(\lambda(x_0))=diag[\lambda(x_0),\tau(\lambda(x_0)),\tau^2(\lambda(x_0))]$
is  2-group decodable \cite[Proposition 7]{R13}. $S(\lambda(x_0))$, $S(\lambda(x_1))$ and $S(\lambda(x_2))$
contain each 6 complex information symbols. By Proposition \ref{prop:groupdeco}, the ML-decoding complexity of the code is
at most $\mathcal{O}(M^{15})$ and the code is fast-decodable. We are no experts in coding theory but assume that hard-limiting the code as done in
\cite{R13} might reduce the ML-complexity further, by a factor of $\sqrt{M}$, to at most $\mathcal{O}(M^{14.5})$.

In comparison, the fully diverse rate-3 VHO-code for 6 transmit and 3 receive antennas presented in \cite[X.C]{VHO} has a complexity of at most
$\mathcal{O}(4 M^{27})$. The fast decodable code rate-3 code for 6 transmit and 3 receive
antennas proposed in \cite[V.B]{MO13} is not fully diverse and has decoding complexity
$\mathcal{O}(M^{30})$.

\subsection{An $8 \times 4$ MIMO System}\label{subsec:5.4}

Let
\begin{enumerate}
\item $\theta = \zeta_{15} + \zeta_{15}^{-1} = 2\cos \frac{2 \pi}{15}$ where $\zeta_{15}$ is a primitive $15^{th}$ root
 of unity and $F = \mathbb{Q}(\theta)$;
\item $K = F(i)$ and $D = (K/F, \sigma, -1)$ which is a subalgebra of Hamilton's quaternions;
\item $L = \mathbb{Q}(i)$ so that $K/L$ is a cyclic field extension of degree 4 with Galois group generated by the
automorphism $\tau: \zeta_{15} + \zeta_{15}^{-1} \mapsto \zeta_{15}^2 + \zeta_{15}^{-2}$;
\item $A = {\rm It}_R^4(D, \tau, i)$.
\end{enumerate}
The associated code is
\[\mathcal{C}_{8 \times 4} = \left \lbrace \left[ \begin{array}{cccc}
\lambda(x_0) & i\lambda(\tl(x_3)) & i\lambda(\tl^2(x_2)) & i \lambda(\tl^3(x_1))\\
\lambda(x_1) & \lambda(\tl(x_0)) & i\lambda(\tl^2(x_3)) & i\lambda(\tl^3(x_2)) \\
\lambda(x_2) & \lambda(\tl(x_1)) & \lambda(\tl^2(x_0)) & i\lambda(\tl^3(x_3))\\
\lambda(x_3) & \lambda(\tl(x_2)) & \lambda(\tl^2(x_1)) & \lambda(\tl^3(x_0))
\end{array}\right] \right\rbrace.\]
If $x_i = a_i + eb_i$ for $a_i, b_i \in K$, then
\[\lambda(x) = \left[ \begin{array}{cc}
a_i & -\sigma(b_i)\\
b_i & \sigma(a_i)
\end{array}\right].\]
 With the encoding from \ref{5.1}, we  encode 32
complex information symbols with each codeword $S$. The code has rate $4$ for 8 transmit and 4 receive antennas
which is maximal.
Assuming $s_j\in \mathbb{Z}[i]$ are $M$-QAM-symbols and
$\{\theta_1,\theta_2,\theta_3, \theta_4\}$ is a basis of the principal ideal in $\mathcal{O}_K$ generated
by $\theta_1=\alpha=1-3i+i\theta^2$ with $\theta_2=\alpha\theta$, $\theta_3=\alpha\theta(-3+\theta^2)$, $\theta_4=
\alpha(-1-3\theta+\theta^2+\theta^3)$. Since all entries of a code matrix $S\in\mathcal{C}_{8\times 4}$ lie
in $\mathcal{O}_K$,
$\det(S)\in\mathcal{O}_L=\mathbb{Z}[i]$ by Lemma \ref{lem:matrix_rep_A} (iii).
 By Proposition \ref{prop:groupdeco} or Corollary \ref{cor:ML}, the ML-decoding complexity of the code is
at most $\mathcal{O}(M^{26})$ and the code is fast-decodable.  Hard-limiting the code as done in
\cite{R13} might reduce the ML-complexity further to $\mathcal{O}(M^{25.5})$.

We have $i \neq z \tl(z) \tl^2(z) \tl^3(z)$
for any $z \in D$ \cite{R13}. We are not able to check whether the code is fully diverse, since
we cannot exclude the
possibility that $F(t)=t^4-i$ decomposes into two irreducible polynomials in $D[t;\widetilde{\tau}^{-1}]$,
we are only able to exclude some obvious cases.

\section{How to design fast-decodable fully diverse  MIMO systems using ${\rm It}^n(D, \tau, d)$ with $d\in F\setminus F_0$ and $n$ prime}

 We assume the set-up from Section \ref{5.1} with the additional condition that
 $n$ is prime and in case $n\not=2,3$, additionally that $F_0$ contains a primitive $n$th root of unity.
In order to construct fully diverse codes, we do not need to restrict our considerations to sparse
codes as done in \cite{SP14}:

\subsection{} \label{6.1}

We assume that
$d\in F\setminus F_0$, such that $d^m\not\in F_0$.
Then $A={\rm It}^n(D, \tau, d)$ is division and each codeword in $\mathcal{C}$ is a matrix of the form given in
 (\ref{equ:main}), which becomes
 (\ref{equ:matrix_rep_A}), as $d\in F$, hence $\lambda(d)={\rm diag}[d,\dots,d]$. These are
 invertible $mn\times mn$ matrices with entries in $K$. $\mathcal{C}$ is fully diverse by Proposition \ref{prop:6}.

\begin{remark}
Suppose that $n$ is an odd prime. If $n\not=3$,  additionally assume that $F_0$ contains a primitive $n$th root of unity.
 Then for $d\in F=F_0(\alpha)$, $d=d_0+d_1\alpha\dots+ d_{n-1}\alpha^{n-1}$ ($d_i\in F_0$),
 it is easy to calculate examples with $d^m\not\in F_0$, e.g. if $n>2$ is a prime and $m=2$,
 any $d=d_0+d_1\alpha$, $d_1\not=0$ works.
\end{remark}

Contrary to \cite{SP14}, we now use the $mn^2$ degrees of freedom of the
channel to transmit $mn^2$ complex information symbols per codeword.
If $mn$ channels are
used, the space-time block code $\mathcal{C}$ consisting of matrices $S$ of the form (\ref{equ:matrix_rep_A})
 with entries as in (\ref{lambdax})
 has a rate of $n$ complex symbols per channel use, which is maximal for $n$ receive antennas.
By Proposition \ref{prop:groupdeco}, which holds analogously,
if the subset of codewords in $\mathcal{C}$ made up of the diagonal block matrix
$$S(\lambda(x_0))=diag [ \lambda(x_0),  \tau(\lambda(x_0))\dots, \tau^{n-1}(\lambda(x_0))]$$
 is $l$-group decodable, then $\mathcal{C}$ has
ML-decoding complexity  $\mathcal{O}(M^{mn^2-mn(l-1)/l)})$ and is fast-decodable.

Suppose that $D={\rm Cay}(K/F,-1)$ is a subalgebra of $\mathbb{H}$.
By Corollary \ref{cor:ML}, which holds analogously, the corresponding code
$\mathcal{C}$ in (\ref{equ:matrix_rep_A})  has
decoding complexity
$$\mathcal{O}(M^{2n^2-3n/2})$$
 if the $s_i$ take values from $M$-QAM and decoding complexity
$$\mathcal{O}(M^{2n^2-n})$$
if the $s_i$ take values from $M$-HEX.
It is fully diverse for all $d\in F\setminus F_0$, such that $d^2\not\in F_0$.

\subsection{Example of $6 \times 3$ MIMO System}
Take the setup of Section \ref{subsec:I} but use  ${\rm It}^3(D, \tau, d)$.
For all $d\in \mathbb{Q}(\theta)\setminus \mathbb{Q}$ with $d^2\not\in \mathbb{Q}$,
${\rm It}^3(D,\tau,d)$ is a division algebra (Proposition \ref{prop:6}). For instance,
${\rm It}^3(D,\tau,\theta)$
is a division algebra and the code $\mathcal{C}$ given by the matrices
\[S=\left[ \begin{array}{cccccc}
x_0 & -\sigma(x_1) &  \theta\tl(z_0) & -\theta\tl \sigma(z_1) & \theta \tl^2(y_0) & -\theta \tl^2 \sigma(y_1)\\
x_1 & \sigma(x_0) &  \theta\tl(z_1) &  \theta\tl \sigma(z_0) & \theta \tl^2 \sigma(y_1) & \theta \tl^2\sigma(y_0)\\
y_0 & -\sigma(y_1) & \tl(x_0) & -\tl \sigma(x_1) & \theta \tl^2(z_0) & -\theta \tl^2\sigma(z_0) \\
y_1 & \sigma(y_0) & \tl(x_1) & \tl\sigma(x_0) & \theta \tl^2(z_1) & \theta \tl^2\sigma(z_1) \\
z_0 & \sigma(z_0) & \tl(y_0) & -\tl\sigma(y_1) & \tl^2(x_0) & -\tl^2\sigma(x_1) \\
z_1 & -\sigma(z_1) & \tl(y_1) & \tl \sigma(y_0) & \tl^2(x_1) & \tl^2 \sigma(x_0)
\end{array} \right]\]
where $x_i, y_i,z_i \in K$, is fully diverse. Using the encoding and notation from Section \ref{5.1},
 for 6 transmit and 3 receive antennas it has
maximal rate 3.

We use $M$-HEX complex constellations (i.e., $s_j\in\mathbb{Z}[\omega]$):
 choose $\{\theta_1,\theta_2,\theta_3\}$ to be a basis of the principal ideal in $\mathcal{O}_K$ generated
by $\theta_1$ with $\theta_1=1+\omega+\theta$, $\theta_2=-1-2\omega+\omega\theta^2$, $\theta_3=(-1-2\omega)+
(1+\omega)\theta+(1+\omega)\theta^2$. Since all entries of the code matrix $S$ lie in $\mathcal{O}_K$,
the determinant of any nonzero codeword $S$ is an element in $\mathcal{O}_F$ by \cite[Theorem 2]{SP14}.

Each codeword
$S(\lambda(x_0))=diag[\lambda(x_0),\tau(\lambda(x_0)),\tau^2(\lambda(x_0))]$
is  2-group decodable \cite[Proposition 7]{R13} and $S(\lambda(x_0))$, $S(\lambda(x_1))$ and $S(\lambda(x_2))$
contain each 6 complex information symbols.
Therefore the ML-decoding complexity of the code is
at most $\mathcal{O}(M^{15})$ and the code is fast-decodable. Again, hard-limiting
 might reduce the ML-complexity to at most $\mathcal{O}(M^{14.5})$.
 The code does not have NVD which would suffice for it to be DMT-optimal. However, NVD seems not always  necessary
 for DMT-optimality to hold.

\section{Conclusion}

One current goal in space-time block coding is to construct space-time block codes which  are fast-decodable
 in the sense of \cite{JR}, \cite{NR}, \cite{LS}
also when there are less receive than transmit antennas, support high data rates
and have the potential to be systematically  built for given numbers of transmit and receive antennas.

After obtaining conditions for the codes associated to the algebras ${\rm It}^n(D, \tau, d)$, $d\in F^\times$,
and ${\rm It}_R^n(D, \tau, d)$, $d\in L\setminus F$, to be fully diverse, we construct
fast decodable fully diverse codes for $mn$ transmit and $n$ receive antennas with maximum rate $n$ out of fast decodable
codes associated with  central simple division algebras of degree $m$, for any choice of $m$ and $n$.
We thus answer the question  for conditions to construct higher rare codes \cite[VII.]{R13}.

The conditions were simplified in the special case of a quaternion algebra
$D$ and an extension $K/L$ with $[K:L]=3$ in Theorem \ref{thm:degree3}, yielding an easy way to construct fully diverse
 rate-3 codes for 6 transmit and 3 receive antennas using ${\rm It}_R^3(D, \tau, d)$, $d\in L\setminus F$.
 They were further simplified for prime $n$ if $n=3$ or if $F_0$ contains a primitive $n$th root of unity
(Proposition \ref{prop:6}),
using ${\rm It}^n(D, \tau, d)$, $d\in F\setminus F_0$ for the code construction.

Since we are dealing with nonassociative algebras and skew polynomial rings, there is no well developed theory of
 valuations or similar yet which one could use to study the algebras over number fields.
 This would go beyond the scope of this paper and will be addressed in \cite{B14}.

\section{Acknowledgments}

We would like to thank the referees for their comments and suggestions which greatly helped  to improve the paper,
and B. Sundar Rajan (Senior Member, IEEE) and L. P. Natarajan for allowing us to include
Lemma \ref{lem:matrix_rep_A}.



\begin{thebibliography}{99}


\bibitem{AP}
 Astier,  V., Pumpl\"{u}n, S.,
\emph{Nonassociative quaternion algebras over rings},
 Israel J. Math. \textbf{155} (2006), 125--147.


\bibitem{B14}
 C. Brown,
 PhD Thesis University of Nottingham, in preparation.

\bibitem{J96}
 N.~Jacobson,
``Finite-dimensional division algebras over fields'',
 Springer Verlag, Berlin-Heidelberg-New York, 1996.



\bibitem{JR}
 G. R.~Jithamitra, B. S.~Rajan,
\emph{Minimizing the complexity of fast-sphere decoding of STBCs},
 IEEE Int. Symposium on Information Theory Proceedings (ISIT), 2011.

\bibitem{KMRT}
 Knus, M.A., Merkurjev, A., Rost, M., Tignol, J.-P.,
 ``The Book of Involutions'',
 AMS Coll. Publications \textbf{44} (1998).


\bibitem{MO13}
 N.~Markin, F.~Oggier,
 \emph{Iterated Space-Time Code Constructions from Cyclic Algebras,}
    IEEE Trans. Inf. Theory \textbf{59} (2013), 5966--5979.


 %
\bibitem{NR}
 L. P.~Natarajan, B. S.~Rajan,
\emph{Fast group-decodable STBCs via codes over GF(4)},
 Proc. IEEE Int. Symp. Inform. Theory, Austin, TX, June 2010
%
\bibitem{LS}
 L. P.~Natarajan and B. S.~Rajan,
\emph{Fast-Group-Decodable STBCs via codes over GF(4): Further Results},
 Proceedings of IEEE ICC 2011, (ICC'11), Kyoto, Japan, June 2011.

\bibitem{NR2}
 L. P.~Natarajan, B. S.~Rajan,
 written communication, 2013.

\bibitem{P66}
 J.-C. Petit,
\emph{Sur certains quasi-corps g\'{e}n\'{e}ralisant un type d'anneau-quotient},
 S\'{e}minaire Dubriel. Alg\`{e}bre et th\'{e}orie des nombres \textbf{20} (1966 - 67), 1--18.


\bibitem{SP14}
  S. Pumpl\"un, A. Steele,
 \emph{Fast-decodable MIDO codes from nonassociative algebras,}
 Int. J. of Information and Coding Theory (IJICOT) \textbf{3}  2015, 15-38.


 \bibitem{P13.2}
   S.~Pumpl\"un,
 \emph{How to obtain division algebras used for fast decodable space-time block codes,}
  Adv.  Math.  Comm. \textbf{8}  (2014), 323--342.

 \bibitem{P14}
  S.~Pumpl\"un,
 \emph{Tensor products of nonassociative cyclic algebras,}
  Online at arXiv:1504.00194[math.RA]

\bibitem{PU11}
 S.~Pumpl\"un,  T.~Unger,
\emph{Space-time block codes from nonassociative division algebras,}
 Adv.  Math.  Comm. \textbf{5}  (2011), 609-629.

\bibitem{SR3}
 K. P.~Srinath, B. S.~Rajan,
\emph{DMT-optimal, low ML-complexity STBC-schemes for asymmetric MIMO systems,}
 2012 IEEE International Symposium on Information Theory Proceedings (ISIT),
 2012 , 3043-3047.

 \bibitem{R13}
  K. P.~Srinath, B. S.~Rajan,
 \emph{Fast-decodable MIDO codes with large coding gain,}
 IEEE Transactions on Information Theory  \textbf{2}  2014, 992--1007.



\bibitem{Sch2}
 R.D. Schafer, ``An introduction to nonassociative algebras'',
 Dover Publ., Inc., New York, 1995.


\bibitem{S12}
 A.~Steele,
\emph{Nonassociative cyclic algebras},
 Israel J. Math. \textbf{200}  (2014),  361--387.


\bibitem{SPO12}
 A.~Steele, S.~Pumpl\"un, F.~Oggier,
\emph{MIDO space-time codes from associative and non-associative cyclic algebras,}
 Information Theory Workshop (ITW) 2012 IEEE (2012), 192--196.




%
\bibitem{VHO}
 R. Vehkalahti, C. Hollanti, F. Oggier,
\emph{Fast-Decodable Asymmetric Space-Time Codes from Division Algebras},
 IEEE Transactions on Information Theory,  \textbf{58},  April 2012.



\bibitem{W}
 W. C.~Waterhouse,
 \emph{Nonassociative quaternion algebras,}
 Algebras Groups Geom. \textbf{4} (1987),  365--378.


\end{thebibliography}
\end{document}